\documentclass[11pt]{article}

\usepackage[margin=1in]{geometry}
\usepackage{amsmath, amssymb, amsthm}
\usepackage{threeparttable}

\theoremstyle{plain}

\newtheorem{theorem}{Theorem}[section]

\newtheorem{assumption}[theorem]{Assumption}

\theoremstyle{definition}

\theoremstyle{remark}

\usepackage{graphicx}
\usepackage{natbib}
\usepackage{setspace}
\usepackage{hyperref}
\usepackage{booktabs}
\usepackage{my-Equations}
\usepackage{float}

\begin{document}

\title{Scalable Structural Estimation of Networked Infrastructure: \\
Exact Decomposition for Localized Coordination}
\author{
    L. Kaili Diamond\thanks{lkdiamond@mines.edu} \\
    \textit{Colorado School of Mines}
    \and
    Ben Gilbert \\
    \textit{Colorado School of Mines}
}
\date{\today}
\maketitle
\begin{abstract}

Interaction effects are often economically central in environments where structural dynamic estimation becomes computationally infeasible. Under fixed group membership and sparse within-group interaction structure, the Bellman operator admits a block-diagonal decomposition that allows high-dimensional dynamic programs to be solved through independent group-level subproblems while preserving the original structural problem exactly. The result applies to a class of dynamic discrete choice models in which interactions are confined within stable local groups and state transitions depend only on within-group conditions. We apply the framework to replacement decisions across 14{,}344 GPU node locations in the Titan supercomputer, where operating environments differ systematically across cage positions. The structural estimates reveal significant spatial coordination: both neighboring failures and recent local replacement activity increase replacement incentives. Accounting for these interaction effects materially shifts predicted replacement timing and reveals significant misoptimization costs in benchmarks that assume conditional independence. More broadly, the results show how exploiting sparsity in interaction structures can make fully structural estimation feasible in large-scale networked systems without relying on simulation-based auxiliary moments or numerical approximation.

\end{abstract}

\section{Introduction}

Large-scale infrastructure systems require continual maintenance and replacement of aging capital assets. Operators collectively spend billions of dollars annually on maintenance and replacement decisions that affect not only asset lifespans, but also system reliability, efficiency, and operating costs. In many environments these decisions are inherently interconnected. Shared dynamics such as local failures, maintenance environments, and operational coordination opportunities can cause replacement incentives for one asset to depend directly on nearby activity within the surrounding infrastructure network.

Modeling these interactions is economically important but computationally
difficult. Structural dynamic discrete choice models (SDDCM) provide a natural
framework for analyzing forward-looking replacement behavior, yet
tractable implementations have typically relied on conditional
independence across assets. Once localized interaction effects enter the
Bellman operator, the dimensionality of the joint state space grows
rapidly with the number of interacting units, rendering conventional
nested fixed point estimation computationally infeasible in large-scale
systems. As a result, the environments in which coordination effects are
most economically important are often precisely those in which fully
structural estimation becomes infeasible.

This paper resolves that tension through an exact decomposition for dynamic replacement models with localized interaction structure. Under empirically meaningful conditions, the Bellman problem separates into independent group-level subproblems while preserving the original value functions, policy functions, and likelihood evaluations exactly. In the joint formulation, the state space grows multiplicatively with the number of interacting units, making value function iteration computationally infeasible in large systems. Under decomposition, computation instead scales with the number and size of groups, restoring tractability for a class of spatial structural models that would otherwise remain computationally infeasible. Because the decomposition operates directly on the Bellman operator, the resulting computational gains extend broadly to estimation methods that repeatedly evaluate the dynamic program, including nested fixed point, MPEC, simulation-based, and hybrid likelihood–moment approaches.

Economically, the decomposition restores coordination effects that are typically excluded from tractable high-dimensional models. Conditional independence frameworks treat actions as isolated events, omitting the influence of shared maintenance environments and local reliability conditions. Incorporating these interaction effects materially changes predicted replacement timing, failure incidence, and maintenance clustering, yielding a richer characterization of coordinated capital replacement behavior. The framework also accommodates persistent heterogeneity in operating environments, allowing otherwise identical assets exposed to different local conditions to exhibit distinct deterioration dynamics within the same localized interaction structure.

Applying the framework to GPU replacement decisions across 14,344 GPU locations organized into 479 local groups in the Titan supercomputer yields two main findings. First, the structural estimates reveal economically meaningful spatial coordination in replacement behavior: both local replacement activity and neighboring failures increase replacement incentives. Second, the spatial model generates materially different replacement dynamics and operational outcomes relative to non-spatial specifications that impose conditional independence.

The framework applies to dynamic decision problems with stable localized interaction structure, including electricity infrastructure, telecommunications networks, transportation fleets, distributed energy systems, and other forms of networked capital. The contribution of the paper is therefore both methodological and economic: it establishes an exact decomposition for spatial dynamic discrete choice problems while showing that economically meaningful coordination effects can be estimated fully structurally in large-scale infrastructure environments without approximation.

\section{Related Literature}\label{sec:lit}

This paper contributes to the literatures on structural dynamic discrete choice models, spatial interaction, computational methods for high-dimensional dynamic programming, and infrastructure maintenance.

\subsection{Structural Estimation of Dynamic Decisions}

The estimation of dynamic discrete choice models originates with \citet{rust1987optimal}, who introduced the nested fixed point (NFXP) algorithm for estimating forward-looking replacement decisions. While subsequent work has reduced computational burdens through mathematical programming with equilibrium constraints (MPEC) \citep{su_judd2012} or conditional choice probability (CCP) estimators \citep{hotz1993conditional}, these approaches generally remain tractable by maintaining conditional independence across agents or locations. Once localized interaction effects enter directly into the Bellman state space, dimensionality grows rapidly with the number of interacting units, making fully structural estimation infeasible in large systems.

This paper also relates to the literature on dynamic maintenance and
replacement decisions under heterogeneous deterioration. Structural
replacement models have been widely used to study capital renewal,
equipment maintenance, and durable asset management in settings where
operating costs evolve over time
\citep{rust1987optimal,benkard2000learning}. Much of this literature focuses on environments in which deterioration is unit-specific and decisions are conditionally independent. In contrast, we focus on infrastructure systems involving localized operational dependencies in which replacement incentives respond to neighboring failures and maintenance activity.

\subsection{Spatial Interaction and Identification}

A separate literature studies spatial dependence and peer effects in economic behavior. \citet{manski1993identification} emphasized the identification challenges created by social and spatial interactions, while \citet{bramoulle2009identification} showed that network structure can generate exclusion restrictions useful for identification. Empirical spatial econometric models typically incorporate such interactions using spatial lag or error specifications \citep{anselin1988spatial,lesage2009introduction}. However, these approaches are primarily static and do not model forward-looking optimization within a structural framework. While dynamic games can incorporate strategic interaction \citep{aguirregabiria2007sequential, bajari2007estimating}, they require solving for equilibrium behavior across agents, leading to substantial complexity. Our approach incorporates localized interaction effects without requiring strategic equilibrium computation across agents. Neighboring conditions affect decisions through the state environment within a fully structural single-agent dynamic programming framework.

\subsection{Numerical Approaches and Exact Decomposition}

Previous work has utilized numerical approximation methods to mitigate the curse of dimensionality in high-dimensional dynamic programming \citep{cai2010dynamic}. Control theory literature has explored exact decomposition under separable dynamics. Notably, \citet{TsakirisTarraf2015} derive algebraic conditions under which problems with linear dynamics can be decomposed into independent subproblems by exploiting invariant subspaces. Our decomposition instead arises from the controlled transition structure of a stochastic dynamic discrete choice model with fixed group membership and localized interactions. Rather than relying on linear dynamics or algebraic state-space transformations, the result exploits separability induced by localized transition dependence across stable groups.

The paper also builds on \citet{diamond2025economics}, which identifies spatial coordination in infrastructure replacement decisions using a hybrid NFXP–MSM framework. That approach incorporated interaction effects through simulation-based moments rather than exact likelihood estimation. The decomposition developed here removes this computational constraint by exploiting the block-diagonal structure of the Bellman operator to preserve the original likelihood problem and heterogeneous baseline specification exactly. As a result, fully structural estimation, likelihood-based inference, and counterfactual analysis become feasible in applications that were previously computationally intractable.

\subsection{Infrastructure Maintenance and Reliability}

Finally, the empirical application relates to the literature on maintenance and reliability in complex engineered systems. Operations research models typically focus on optimal policy design under known failure processes, rather than estimating structural parameters from observed behavior. The application also connects to work on hardware reliability in high-performance computing, specifically studies of the Titan supercomputer documenting systematic spatial variation in GPU failure patterns \citep{tiwari2015,ostrouchov2020gpu,min2025spatially}. This literature characterizes important reliability outcomes, but does not model the economic decision problem governing replacement behavior. The framework developed here integrates these reliability patterns into a structural model of forward-looking infrastructure maintenance, bridging the gap between engineering reliability and structural econometrics.

This paper contributes to these literatures by establishing an exact decomposition result for structural dynamic discrete choice models with localized interaction structure and fixed group membership. Under these conditions, the Bellman operator separates into independent group-level subproblems while preserving value functions, policy functions, and likelihood evaluations exactly. The decomposition therefore expands the class of spatial models that can be estimated fully structurally in practice.

\section{Structural Framework}\label{sec:model}

This section develops a structural dynamic replacement model with localized interaction effects across assets. The environment consists of stable local groups in which replacement incentives depend both on own-unit deterioration and local operating conditions. Fixed group membership implies that transition dynamics depend only on within-group conditions rather than on the full joint system state. The model is developed in general form so that the decomposition result extends broadly to dynamic systems with fixed local interaction structure.

\subsection{Economic Environment}

Consider a collection of infrastructure assets indexed by
$i = 1,\dots,N$ observed over discrete time periods
$t = 1,\dots,T$. In each period, operators choose whether to continue
operating an existing asset or replace it with a new unit. These
decisions are forward-looking because current replacement choices affect future operating conditions, reliability, and replacement incentives.

Assets often operate within localized physical or operational environments, such as servers within data-center racks, transformers within distribution circuits, or vehicles within maintenance depots. In these settings, replacement incentives depend not only on own-unit deterioration, but also on nearby operating conditions and maintenance activity arising from shared infrastructure, downtime costs, or maintenance constraints.

A central feature of these environments is that interactions are local rather than global. Assets belong to stable operational groups within which maintenance access, environmental exposure, and coordination opportunities are shared. Group membership remains fixed over time, implying that state transitions depend on local conditions within the group rather than on the full system state.

The framework incorporates interaction effects through the state environment rather than through explicit equilibrium interactions across agents. The key methodological contribution is to show that when group membership is fixed and transitions depend only on local conditions, the resulting dynamic program decomposes exactly into independent group-level problems.

\subsection{States, Groups, and Local Interaction Structure}

Each asset $i$ belongs to a fixed group $g(i) \in \mathcal{G}$, where groups represent shared conditions. Groups may correspond to physical neighborhoods, shared environments, maintenance regions, or any stable organizational structures governing local interaction. Fixed group membership arises naturally when physical location, maintenance access, or network topology remain stable over the planning horizon.

When group membership is assumed to remain fixed over time, state transitions for asset $i$ depend only on own-unit conditions and
local interaction variables generated within group $g(i)$. This fixed
group structure induces a sparse transition architecture in which
dependence is confined within groups rather than spanning the full
system.

Within each group, let $\mathcal{N}(i)$ denote the interaction neighborhood of asset $i$.
The neighborhood contains those assets whose operating conditions or
maintenance activity may affect the replacement incentives of asset
$i$. Interactions are assumed to be localized within groups, so that $\mathcal{N}(i) \subseteq g(i)$.

\newpage

At the beginning of period $t$, operators observe the state

\begin{equation}
s_{it}
=
(a_{it},\, c_i,\, f_{it},\, n^d_{i,t-1},\, n^f_{it}),
\end{equation}

where:

\begin{itemize}
    \item $a_{it}$ denotes the asset deterioration state, such as age,
    cumulative wear, or operating intensity,
    
    \item $c_i$ indexes the fixed depreciation state $i$,
    
    \item $f_{it}$ captures realized own-unit failure or operating conditions in period $t$,
    
    \item $n^d_{i,t-1}$ summarizes recent maintenance or replacement activity within the local neighborhood,

    \item $n^f_{it}$ summarizes contemporaneous local failure or operating shocks.
\end{itemize}

The fixed depreciation state $c_i$ allows otherwise similar assets to face different deterioration or reliability dynamics. For example, $c_i$ may represent climate exposure, wildfire risk, utilization intensity, or service environment. This state is distinct from the interaction group $g(i)$, which defines the fixed group structure governing localized interactions and the decomposition of the dynamic program. 

Lagged neighborhood activity $n^d_{i,t-1}$ is predetermined with respect to current-period decisions, while contemporaneous local shocks $n^f_{it}$ are treated as observed prior to choice. Consequently, local operating conditions enter the state before decisions are made.

This localized transition structure induces the block separability underlying the decomposition result developed below.

\subsection{Replacement Decisions and Utility}

In each period, operators observe the current state and idiosyncratic utility shocks before choosing an action $d_{it} \in \{0,1\}$, where

\[
d_{it}
=
\begin{cases}
0 & \text{continue operation}, \\
1 & \text{replace asset}.
\end{cases}
\]

The per-period utility from continued operation is

\begin{equation}
u(s_{it},0;\theta,\gamma)
=
\delta_{g_i} a_{it}
+
\theta_f f_{it}
+
\gamma_{lag} n^d_{i,t-1}
+
\gamma_{fail} n^f_{it},
\end{equation}

while replacement yields

\begin{equation}
u(s_{it},1;\theta,\gamma)
=
\theta_{RC}.
\end{equation}

The parameter vector $\theta$ governs own-state operating incentives,
while $\gamma = (\gamma_{lag}, \gamma_{fail})$ captures localized
interaction effects.

The coefficient $\gamma_{lag}$ measures how recent neighboring replacement activity affects continuation values, while $\gamma_{fail}$ captures the effect of nearby failures. Negative interaction coefficients imply that recent local maintenance activity or neighboring failures increase replacement incentives, consistent with coordinated maintenance behavior, shared setup costs, sequential replacement dynamics, or local reliability spillovers.

Environment-specific deterioration parameters $\delta_{gci}$ allow identical assets operating in different local environments to exhibit different depreciation or reliability dynamics.

\subsection{Dynamic Optimization Problem}

Operators are assumed to be forward-looking and maximize the expected
discounted value of utility over time. Let
$\beta \in (0,1)$ denote the discount factor.

Conditional on the observed state, the value function satisfies the
Bellman equation

\begin{equation}
V(s)
=
\max_{d \in \{0,1\}}
\left\{
u(s,d;\theta,\gamma)
+
\beta
\sum_{s'}
P(s' \mid s,d)
V(s')
\right\}.
\label{eq:bellman-general}
\end{equation}

Assuming additive Type-I extreme value shocks to utility yields the
standard logit representation for the integrated value function

\begin{equation}
\overline{V}(s)
=
\log
\sum_{d \in \{0,1\}}
\exp
\left[
u(s,d;\theta,\gamma)
+
\beta
\sum_{s'}
P(s' \mid s,d)
\overline{V}(s')
\right].
\label{eq:bellman-spatial}
\end{equation}

Because interaction variables enter the joint state space directly, conventional solution methods quickly become computationally infeasible in large systems.

The next section shows that fixed group membership and localized transition dependence induce a block-diagonal structure in the controlled transition kernel, allowing the Bellman problem to decompose into independent group-level subproblems.
    
\section{Bellman Decomposition}\label{sec:decomposition}

Infrastructure systems often exhibit stable localized interaction structure. Assets are frequently organized into persistent groups that share operating conditions, maintenance environments, or reliability exposure. Under these conditions, replacement dynamics depend primarily on local states rather than on the full system state.

\subsection{Localized Interaction Structure}

Figure~\ref{fig:decomposition} illustrates the logic of the decomposition. Rather than solving the Bellman problem over the global state space, the model can be solved through independent group-level subproblems. Under fixed group membership and localized transition dependence, combining these solutions reproduces the same value and policy functions as the joint problem.

\begin{figure}[H]
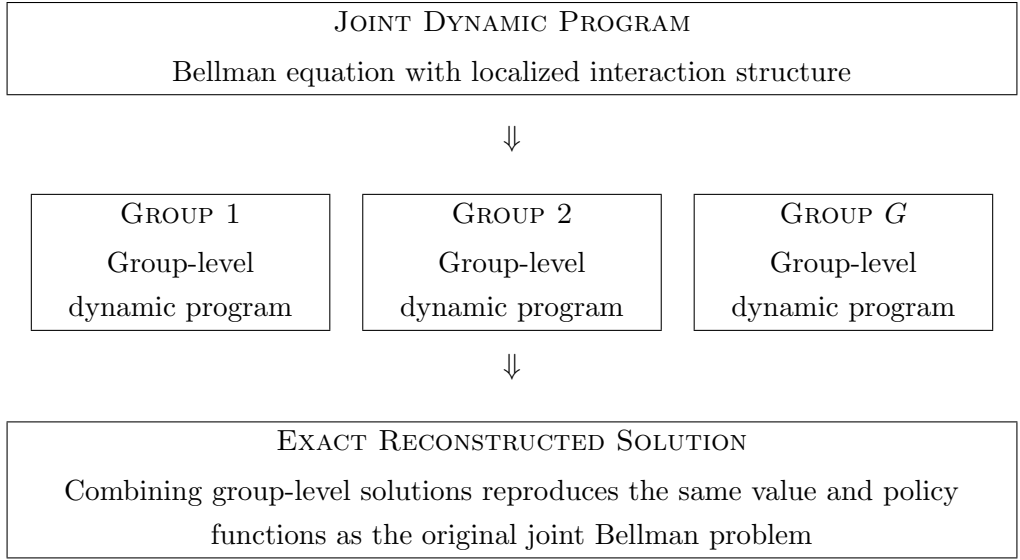

\centering

\renewcommand{\arraystretch}{1.4}

\begin{tabular}{c}

\fbox{
\parbox{0.78\textwidth}{
\centering
\textsc{Joint Dynamic Program}\\[0.2em]
Bellman equation with localized interaction structure
}
}

\\[1em]

$\Downarrow$

\\[1em]

\begin{tabular}{ccc}

\fbox{
\parbox{0.21\textwidth}{
\centering
\textsc{Group 1}\\[0.2em]
Group-level dynamic program
}
}
&
\fbox{
\parbox{0.21\textwidth}{
\centering
\textsc{Group 2}\\[0.2em]
Group-level dynamic program
}
}
&
\fbox{
\parbox{0.21\textwidth}{
\centering
\textsc{Group $G$}\\[0.2em]
Group-level dynamic program
}
}

\end{tabular}

\\[1em]

$\Downarrow$

\\[1em]

\fbox{
\parbox{0.78\textwidth}{
\centering
\textsc{Exact Reconstructed Solution}\\[0.2em]
Combining group-level solutions reproduces the same value and policy
functions as the original joint Bellman problem
}
}

\end{tabular}

\caption{Exact decomposition of the joint Bellman problem into independent group-level dynamic programs under fixed group membership and localized transition dependence.}
\label{fig:decomposition}

\end{figure}

Appendix~\ref{app:matrix-decomposition} provides a simple matrix illustration of how the decomposition removes only off-diagonal zero-block computations while preserving the original joint Bellman problem exactly.

\subsection{Exact Bellman Decomposition}

The following assumptions characterize environments in which the Bellman operator admits an exact decomposition into independent group-level subproblems. The result relies on fixed group membership, localized transition dependence, and additive separability of per-period utility.

\begin{assumption}[Fixed Group Membership]
  \label{asm:fixed-groups}
  Each asset location $i$ belongs to exactly one group $g(i) \in \mathcal{G}$
  for all periods $t$.
\end{assumption}

\begin{assumption}[Localized Transition Dependence]
  \label{asm:cond-indep}
  The transition kernel of group $g$ depends only on the state and
  actions within that group:
  \[
    P(\mathbf{s}_g' \mid \mathbf{s}, \mathbf{d}) =
    P(\mathbf{s}_g' \mid \mathbf{s}_g, \mathbf{d}_g).
  \]
\end{assumption}

\begin{assumption}[Additive Separability]
  \label{asm:separable}
  Per-period utility is additively separable across groups:
  \[
    u(\mathbf{s},\mathbf{d}) = \sum_{g=1}^{G} u_g(\mathbf{s}_g,\mathbf{d}_g).
  \]
\end{assumption}

These conditions imply that both the payoff function and the law of motion factor across groups, which is sufficient for separability of the Bellman operator.

\begin{theorem}[Exact Bellman Decomposition]
\label{thm:decomposition}

Under Assumptions~\ref{asm:fixed-groups}--\ref{asm:separable},
the Bellman operator for the joint dynamic program decomposes across
groups. The value and policy functions of the full system can be obtained
by solving a collection of independent group-level Bellman problems.

Consequently, under the maintained assumptions the decomposed and joint formulations yield identical value functions, policy functions, and likelihood evaluations.
\end{theorem}

\begin{proof}
Let $\mathbf{s} = (\mathbf{s}_1,\ldots,\mathbf{s}_G)$ denote the joint
state vector, where $\mathbf{s}_g$ collects the states of assets in
group $g$. Let $\mathbf{d} = (\mathbf{d}_1,\ldots,\mathbf{d}_G)$
denote the vector of actions.

The Bellman operator for the joint dynamic program is
\[
  \mathcal{T}V(\mathbf{s})
  =
  \max_{\mathbf{d}}
  \Bigg\{
    u(\mathbf{s},\mathbf{d})
    +
    \beta
    \sum_{\mathbf{s}'}
    P(\mathbf{s}' \mid \mathbf{s},\mathbf{d})\,V(\mathbf{s}')
  \Bigg\}.
\]

Under Assumption~\ref{asm:separable}, utility is additively separable:
\[
  u(\mathbf{s},\mathbf{d}) = \sum_{g=1}^{G} u_g(\mathbf{s}_g,\mathbf{d}_g).
\]

Under Assumption~\ref{asm:fixed-groups}, assets remain in the same
group over time. Under Assumption~\ref{asm:cond-indep}, the transition
kernel factorizes across groups:
\[
  P(\mathbf{s}' \mid \mathbf{s},\mathbf{d})
  =
  \prod_{g=1}^{G}
  P_g(\mathbf{s}'_g \mid \mathbf{s}_g,\mathbf{d}_g).
\]

Separability of utility, transitions, and feasible actions implies that the Bellman operator can be written as a collection of group-level operators
\[
  \mathcal{T}_g V_g(\mathbf{s}_g)
  =
  \max_{\mathbf{d}_g}
  \left\{
    u_g(\mathbf{s}_g,\mathbf{d}_g)
    +
    \beta
    \sum_{\mathbf{s}'_g}
    P_g(\mathbf{s}'_g \mid \mathbf{s}_g,\mathbf{d}_g)\,
    V_g(\mathbf{s}'_g)
  \right\}.
\]

The resulting fixed point $V_g^*$ therefore inherits the same additive structure.

\[
  V^*(\mathbf{s}) = \sum_{g=1}^{G} V_g^*(\mathbf{s}_g).
\]

Thus solving the collection of group-level Bellman problems produces
the same fixed point, and therefore the same value and policy
functions, as solving the full joint Bellman equation.

\end{proof}

This equivalence implies that likelihood evaluation and policy simulation can be conducted using the decomposed representation while preserving the original joint likelihood problem exactly. The decomposition does not imply separate estimation across groups or a pooling of parameters. Rather, it provides a computationally feasible way to estimate the full structural model using the joint behavior of all units. Consequently, the decomposition is not a numerical approximation or reduction in model scope, but an exact representation of the same economic problem under empirically meaningful restrictions.


\subsection{Computational Implications}


Consider a dynamic program satisfying Assumptions~\ref{asm:fixed-groups}--\ref{asm:separable}. In the joint formulation, Bellman iteration scales with the full joint state space, whose dimensionality grows multiplicatively with the number of interacting units. Under exact decomposition, the problem instead separates into independent group-level Bellman operators.

Consequently, computation scales with the number and size of groups rather than the full joint state space. In the joint formulation, Bellman iteration scales as $\mathcal{O}\!\left((\prod_{g=1}^{G} |\mathcal{S}_g|)^2\right)$. Under decomposition, the problem reduces to $\mathcal{O}\!\left(\sum_{g=1}^{G} |\mathcal{S}_g|^2\right)$ group-level computations.

Because the decomposition operates directly on the Bellman operator, the computational gains extend to estimation methods that require repeated solution of the dynamic program, including nested fixed point (NFXP), mathematical programming with equilibrium constraints (MPEC), simulation-based estimators, and hybrid likelihood–moment methods.

These gains are important not only for estimation, but also for policy analysis and inference. Bootstrap inference, robust standard errors, counterfactual simulation, sensitivity analysis, and model comparison all require repeated evaluation of the structural model. By reducing the cost of the Bellman solution, the decomposition makes these procedures operationally feasible at the scale of modern infrastructure systems within a fully structural framework.
\section{Empirical Application: Titan Supercomputer}

\label{sec:application}
The assumptions underlying Theorem~\ref{thm:decomposition} map naturally into the Titan application. Assets operate within stable cage environments that generate persistent heterogeneity in operating conditions and localized interaction structure. The empirical specification therefore defines interaction variables within cages and restricts state transitions to depend on local operating conditions rather than the full system state. These conditions align closely with the institutional structure of the Titan supercomputer, whose administrative records provide detailed information on failures, maintenance activity, and replacement behavior across GPU locations.

\subsection{Institutional Environment}

The empirical analysis uses administrative records from the Titan supercomputer operated at Oak Ridge National Laboratory (ORNL) between 2012 and 2019. Titan was a large-scale GPU-accelerated computing system with approximately 19{,}200 node positions representing potential GPU installation locations. These GPUs served as the primary computational units for scientific workloads and constitute the capital equipment studied in this analysis.

Titan provides a useful setting for studying infrastructure replacement because it combines large-scale capital deployment with stable physical layout and detailed operational records. GPUs operate continuously under heterogeneous thermal conditions, experience nontrivial failure risk over time, and are maintained within localized maintenance environments. These features generate both persistent heterogeneity in deterioration dynamics and opportunities for coordinated maintenance activity.

The system’s physical architecture defines a natural hierarchal spatial asset structure.
GPUs are installed within cabinets, and each cabinet contains three
vertically stacked cages. Each cage houses up to 32 GPUs that share
cooling, power, and network infrastructure. The unit of observation is a
node, representing a physical GPU location.

The facility architecture follows a nested hierarchy:
\begin{center}
\textbf{Facility} $\rightarrow$ \textbf{Cabinet} $\rightarrow$ \textbf{Cage} $\rightarrow$ \textbf{Node}
\end{center}

In the empirical analysis, the node defines the decision unit, while the
cage defines the local operating environment. GPUs within a cage share
common thermal conditions and maintenance access, making the cage the
relevant level at which local interactions arise. Maintenance operations are also naturally organized at this level. The localized service structure creates opportunities for coordinated maintenance and replacement activity within cages.

A key feature of Titan is a vertical thermal gradient across cages within
a cabinet. Due to floor-up cooling, GPUs in higher cages operate at
systematically higher temperatures than those in lower cages. Therefore,
cages are classified as $\text{cage} \in \{0,1,2\}$ corresponding to bottom
(coolest), middle, and top (hottest) positions. This creates persistent
heterogeneity in operating conditions across otherwise identical
hardware. Because GPUs are not randomly reassigned across cages during operation, these thermal environments remain persistent over time. The resulting variation generates systematic differences in deterioration and failure risk across otherwise homogeneous assets.

Several features of the system support identification. GPUs were deployed
with homogeneous hardware specifications, eliminating baseline quality
differences. Locations are fixed by the physical layout of the system,
so thermal environments are exogenously assigned. Administrative records
provide detailed information on installation, removal, and failure events
throughout the lifetime of the system. The records additionally contain precise location identifiers linking each GPU to its cabinet and cage position, allowing replacement and failure behavior to be analyzed within stable local neighborhoods over time.

These features create an environment in which identical assets operate under systematically different local conditions within fixed groups, making Titan a useful setting for analyzing interaction effects in dynamic maintenance decisions.

\subsection{Panel Construction}

The raw administrative logs record installation events, failure events,
and removals at the GPU level. These records are converted into a monthly
panel tracking the state of each GPU location over time. These records are linked across serial numbers and physical locations to
construct a longitudinal history of equipment operation and replacement
behavior at the node level. Monthly aggregation preserves the timing of replacement and failure
dynamics relevant for forward-looking maintenance decisions while
smoothing high-frequency operational noise.

Following \citet{ostrouchov2020gpu}, replacement events associated with manufacturer defects and coordinated warranty refresh cycles are excluded from the analysis. These episodes reflect centralized interventions rather than decentralized operational decisions. The final decommissioning phase of the system is also excluded.

Additional cleaning ensures a consistent mapping between locations and
replacement decisions. Overlapping installation and removal records are
resolved so that each location contains at most one GPU per period.
GPUs that are relocated and reused as backfill units are restricted to
their first observed stint, ensuring that replacement decisions are
defined consistently at the location level.

For estimation, the sample is further restricted to locations with complete coverage over a common analysis window, yielding 378{,}124 location-month observations across 14{,}344 observed GPU locations organized into 479 cage-level groups.

\subsection{Scope and Empirical Limitations}

The Titan application provides a valuable empirical setting for studying
localized replacement behavior because it combines detailed
administrative records with stable physical infrastructure and observable
maintenance activity at scale. Such data are rarely available outside
large institutional or government-operated systems, particularly at the
node level required for structural estimation of localized interaction
effects.

The empirical results should be interpreted as evidence from a specific institutional environment rather than as a
complete characterization of modern computing infrastructure. Titan represents a single large-scale system operating during
a particular technological generation, and the hardware environment
differs substantially from contemporary hyperscale AI deployments.

In addition, Titan approached decommissioning during the estimation
window, which likely weakened the empirical relevance of stronger
forward-looking Markov dynamics. Consequently, operators may have placed
less weight on long-run continuation values than would be expected in a
stationary environment with an indefinite planning horizon. As a result,
the estimated dynamics should not be interpreted as fully representative
of replacement behavior in perpetually operating infrastructure systems.

Nevertheless, Titan provides an important proof-of-concept environment for studying localized coordination in large-scale infrastructure systems. The empirical application demonstrates that sufficiently detailed operational data can support fully structural estimation of localized interaction effects.

\subsection{Empirical Specification}\label{sec:empirical}

This section describes the empirical specification used to estimate the structural model in the Titan application. The discussion introduces the state variables, transition process, neighborhood structure, and timing assumptions.

\subsubsection{Empirical Neighborhood Structure}

The cabinet-and-cage structure defines the relevant local neighborhood for each unit. For a GPU located at node $i$, the neighborhood consists of other GPUs within the same cage of the same cabinet. This structure captures both localized operating conditions and the feasible scope of maintenance coordination. The resulting neighborhood variables therefore reflect both correlated environmental exposure and localized maintenance organization.

\begin{samepage}
Two variables summarize local neighborhood conditions. First, lagged replacement activity is defined as
\begin{align}
n_{i,t-1}^{d} = \mathbf{1}\{\text{any other GPU in the same cage was replaced at } t-1\}.
\end{align}
This captures sequential replacement behavior following recent
maintenance activity.
\end{samepage}
Second, contemporaneous failure activity is defined as
\begin{align}
n_{it}^{f} = \mathbf{1}\{\text{any other GPU in the same cage failed at } t\}.
\end{align}
This captures local reliability shocks that may affect replacement
incentives.

Lagged replacement activity is predetermined with respect to the
decision at time $t$, while failure realizations are assumed to be
observed prior to the replacement decision within the period. Together,
these variables summarize the local conditions faced by each unit while
preserving a clear temporal ordering between information and choice.

\subsubsection{Empirical State Space}

The empirical state for GPU location $i$ in month $t$ is

\[
s_{it}
=
(a_{it},\, c_i,\, f_{it},\, n^d_{i,t-1},\, n^f_{it}),
\]

where $a_{it}$ denotes GPU age in months, $c_i \in \{0,1,2\}$ indexes cage position within the cabinet, $f_{it}$ indicates a failure event, and the neighborhood variables summarize local replacement and failure conditions within the same cage.

\subsubsection{Transition Structure and Timing}

The model evolves in discrete monthly periods. Within each period, operators observe the current operating state of each asset and local neighborhood conditions before choosing whether to continue operation or replace the asset.

The within-period timing proceeds as follows:

\begin{enumerate}
\item Current-period operating states and failure realizations are observed.
\item Operators observe local neighborhood conditions, including nearby failures and prior period replacement activity within the same cage.
\item Replacement decisions are made.
\item State variables evolve according to the controlled transition process.
\end{enumerate}

The transition kernel therefore combines deterministic and stochastic
components. Age evolves deterministically conditional on replacement decisions, while failure events evolve stochastically conditional on the current state. Transition probabilities are estimated directly from the data.

Neighborhood interaction variables evolve endogenously through observed current-period failure activity and lagged replacement activity within the local interaction group. Contemporaneous failure measures exclude the focal unit in order to separate local operating conditions from own-unit failure realizations.

Because interaction neighborhoods are localized within fixed groups, future neighborhood states depend only on transitions occurring within the same group.

Transition probabilities are estimated nonparametrically from monthly state transitions in the panel data. The resulting transition kernel reflects the empirical evolution of deterioration, failures, and local neighborhood activity within the infrastructure system.

\section{Estimation Procedure}\label{sec:estimation4}

The model is estimated by maximum likelihood using a nested fixed
point (NFXP) algorithm. For a given parameter vector
$(\theta,\gamma)$, the Bellman equation in
\eqref{eq:bellman-spatial} is solved to obtain the expected value
function $\EV(\st)$.

Conditional choice probabilities (CCPs) follow from the Type-I extreme
value assumption:
\begin{equation}
  P(d \mid \st;\theta,\gamma)
  =
  \frac{
    \exp\!\left[
      u(\st,d;\theta,\gamma)
      +
      \beta \sum_{s'} P(s' \mid \st,d)\,\EV(s')
    \right]
  }{
    \sum_{d' \in \{0,1\}}
    \exp\!\left[
      u(\st,d';\theta,\gamma)
      +
      \beta \sum_{s'} P(s' \mid \st,d')\,\EV(s')
    \right]
  }.
\end{equation}

The likelihood contribution for observation $(i,t)$ is given by the
probability of the observed action $d_{it}$ conditional on the state
$\st$. The log-likelihood function is therefore
\begin{equation}
  \ell(\theta,\gamma)
  =
  \sum_{i,t}
  \log P(d_{it} \mid \st;\theta,\gamma).
\end{equation}

A key feature of the estimation procedure is that the Bellman equation
is solved using the decomposition result established in
Section~\ref{sec:decomposition}. Rather than solving a single dynamic
program on the full joint state space, the value function is computed
as a collection of independent group-level solutions. These group-level
value functions are then combined to construct the CCPs used in the
likelihood.

Parameter estimates are obtained by maximizing the log-likelihood
function with respect to $(\theta,\gamma)$ using a derivative-free
optimization routine. At each parameter evaluation, the Bellman
equation is solved to convergence using value function iteration.

Standard errors are computed using two approaches. Asymptotic standard errors are obtained from the inverse Hessian of the objective function at the optimum. A block bootstrap procedure resamples cages, preserving the within-group dependence structure induced by spatial interactions.

\section{Results}
\label{sec:results4}

This section presents the parameter estimates and computational implications, evaluates evidence for localized coordination through likelihood-based inference and specification tests, and examines the robustness of the interaction effects across alternative specifications. The results show that replacement decisions respond systematically to nearby maintenance and failure activity, providing strong evidence of localized coordination in infrastructure replacement behavior.

\subsection{Structural Estimates}

Table~\ref{tab:spatial_structural_params} reports parameter estimates
for the baseline model and the decomposed structural model.

\begin{table}[H]
\centering
\caption{Baseline and Structural Spatial Model Estimates}
\label{tab:spatial_structural_params}
\begin{tabular}{lcc}
\toprule
Parameter & Baseline & Structural Spatial \\
\midrule
Age $\times$ Cage 0 ($\theta_{\text{0}}$) & -0.0063 & -0.0060 \\
Age $\times$ Cage 1 ($\theta_{\text{1}}$) & -0.0089 & -0.0085 \\
Age $\times$ Cage 2 ($\theta_{\text{2}}$) & -0.0198 & -0.0183 \\
Failure Event ($\theta_{\text{fail}}$)  & -8.7718 & -8.7453 \\
Replacement Cost ($\theta_{\text{RC}}$) & -9.3384 & -9.3352 \\
Spatial lag effect ($\gamma_{lag}$) & --- & -0.4314 \\
Spatial failure effect ($\gamma_{fail}$) & --- & -0.4184 \\
\midrule
Negative log-likelihood & 8955.90 & 8886.75 \\
AIC    & 17921.81 & 17787.49 \\
BIC  & 17976.02 & 17863.39 \\
McFadden pseudo-$R^2$ & 0.591 & 0.594 \\
Observations            & 378,124 & 378,124 \\
States                  & 354 & 1,416 \\
\bottomrule
\end{tabular}
\end{table}

The estimated spatial interaction coefficients are economically large
and statistically significant. The coefficient on lagged neighboring
replacement activity, $\hat{\gamma}_{lag} = -0.4314$,
implies that recent nearby replacement activity reduces the relative
utility of continued operation and increases the probability of
replacement. Similarly,
$\hat{\gamma}_{fail} = -0.4184$
indicates that neighboring failure realizations also increase
replacement incentives within the local operating environment.

The age coefficients are negative across all cage environments, implying that continuation values decline as GPUs age. The magnitude of deterioration differs systematically across cage positions, with the hottest environments (cage 2) exhibiting substantially steeper depreciation effects than cooler operating environments (cage 0). This pattern is consistent with persistent heterogeneity in thermal operating conditions across the system architecture. Failure realizations also strongly reduce continuation values, while the replacement cost parameter reflects the substantial fixed cost associated with replacement activity.

The interaction parameters are also negative, implying that both local failures and prior period replacement activity increase current replacement incentives. These estimates provide evidence that replacement incentives respond systematically to nearby operating and maintenance conditions. Operationally, these interaction effects may reflect shared maintenance access, common operating environments, or economies of scale in clustered hardware servicing.

The spatial model improves the negative log-likelihood from 8955.90 to 8886.746 over 378,124 observations while introducing only two additional interaction parameters. The improvement in fit indicates that nearby replacement and failure activity contain economically meaningful information about forward-looking replacement behavior beyond own-unit deterioration alone. The results further suggest that localized operating conditions and maintenance activity play an important role in shaping infrastructure replacement incentives.

Likelihood-based model selection criteria also improve substantially. Relative to the non-spatial specification, the spatial model lowers AIC by more than 130 points and BIC by more than 110 points despite the additional interaction parameters. These improvements provide strong evidence that localized interaction effects contribute systematically to replacement behavior.

The increase in McFadden pseudo-$R^2$ is more modest in absolute terms, which is typical in structural dynamic discrete choice environments where persistent own-state deterioration already explains much of the baseline replacement behavior. The remaining improvement therefore reflects additional explanatory power arising from localized coordination effects.

More importantly, these interaction effects are estimated directly within the Bellman equation through full maximum likelihood estimation rather than through auxiliary simulation or reduced-form approximation. At the scale of the Titan system, solving the corresponding joint spatial dynamic program without decomposition would be computationally infeasible. 

\subsubsection{Empirical Computational Performance}

The decomposition makes fully structural estimation feasible for the
1,416-state spatial model studied in the Titan application.
Table~\ref{tab:runtime} compares estimation runtimes for the hybrid
NFXP--MSM estimator of \citet{diamond2025economics} and the fully
structural likelihood-based estimator developed here.

\begin{table}[H]
\centering
\caption{Estimation runtime comparison}
\label{tab:runtime}
\begin{tabular}{lccc}
\toprule
Estimator & Runtime & Iterations & s/eval \\
\midrule
Hybrid NFXP--MSM & 1{,}139 seconds & 276 & 4.13 \\
Structural NFXP (decomposed) & 39 seconds & 772 & 0.05 \\
\bottomrule
\end{tabular}
\end{table}

Although the fully structural estimator requires more optimizer iterations, each evaluation is substantially faster because it avoids the forward simulation required by the MSM objective. Both estimators use the same decomposed Bellman solution, so the runtime difference reflects the additional simulation burden imposed by the MSM criterion rather than differences in the dynamic programming step itself.

The computational savings are especially important for procedures that require repeated model solution. Although counterfactual analysis and inference are theoretically possible under simulation-based estimators such as NFXP–MSM, repeated model solution and simulation at Titan scale become operationally costly in practice. The decomposed likelihood-based framework substantially reduces this burden, making bootstrap inference, large-scale counterfactual analysis, and repeated policy simulation computationally feasible.

\subsection{Inference and Identification}

The precision of these estimates is confirmed by both asymptotic and bootstrap standard errors, reported in Table~\ref{tab:spatial_structural_se}. Bootstrap standard errors are computed using 100 resampled estimation runs. Despite the high dimensionality of the state space (1,416 states), the spatial interaction coefficients remain stable and statistically significant at conventional levels across both error specifications.

\begin{table}[!h]
\centering
\caption{Structural spatial replacement model estimates and standard errors}
\label{tab:spatial_structural_se}
\begin{threeparttable}
\begin{tabular}{lccc}
\toprule
Parameter & Estimate & Asymptotic SE & Bootstrap SE \\
\midrule
Age $\times$ Cage 0 ($\theta_{\text{0}}$) & $-0.0060$*** & $0.0011$ & $0.0011$ \\
Age $\times$ Cage 1 ($\theta_{\text{1}}$) & $-0.0085$*** & $0.0011$ & $0.0010$ \\
Age $\times$ Cage 2 ($\theta_{\text{2}}$) & $-0.0183$*** & $0.0015$ & $0.0016$ \\
Failure Event ($\theta_{\text{fail}}$)           & $-8.7453$***   & $0.1917$   & $0.3866$   \\
Replacement Cost ($\theta_{\text{RC}}$)         & $-9.3352$***   & $0.4537$   & $0.4378$   \\
Spatial lag ($\gamma_{lag}$)  & $-0.4314$** & $0.1481$ & $0.1748$ \\
Spatial failure ($\gamma_{fail}$) & $-0.4184$*** & $0.1528$ & $0.1455$ \\
\bottomrule
\end{tabular}

\begin{tablenotes}
\small
\item \textit{Notes:} Inference uses bootstrap standard errors from 100 resampled runs.
\item *** $p<0.01$, ** $p<0.05$, * $p<0.10$.
\end{tablenotes}
\end{threeparttable}
\end{table}

While the failure cost standard error nearly doubles under the bootstrap (0.1917 to 0.3866), this likely reflects the relative rarity of failure events. A bootstrap that resamples by cage may occasionally omit critical failure clusters, leading to higher variance in that specific parameter. However, the consistency of the spatial parameters across methods implies that the coordination effect is a global feature of the data and is not driven by outlier cages or inferential framework.

Bootstrap standard errors are generally similar to asymptotic standard errors, indicating that the spatial interaction effects are not driven by a small number of cages or local failure clusters.

\subsection{Likelihood-Based Identification and Model Comparison}

The spatial model is tested against the baseline via a likelihood ratio
statistic:
\begin{equation}
  \mathrm{LR}
  =
  -2\left[
    \ell(\hat{\theta}_{\mathrm{baseline}})
    -
    \ell(\hat{\theta}_{\mathrm{spatial}}, \hat{\gamma})
  \right]
  \sim \chi^2(2)
  \quad \text{under } H_0\colon (\gamma_{lag},\gamma_{fail})=(0,0).
  \label{eq:lr-test}
\end{equation}

The spatial specification decisively rejects the null of no localized interaction effects. The ratio statistic comparing the spatial and non-spatial models is 138.31 ($p < 0.001$), indicating that neighboring replacement and failure activity contain substantial information about replacement incentives beyond own-unit deterioration alone.


Figure~\ref{fig:likelihood_surface} illustrates the local geometry of the
likelihood function around the spatial interaction parameters
$(\gamma_{lag}, \gamma_{fail})$, holding the remaining parameters fixed
at their structural estimates.

\begin{figure}[H]
\centering
\includegraphics[width=\textwidth]{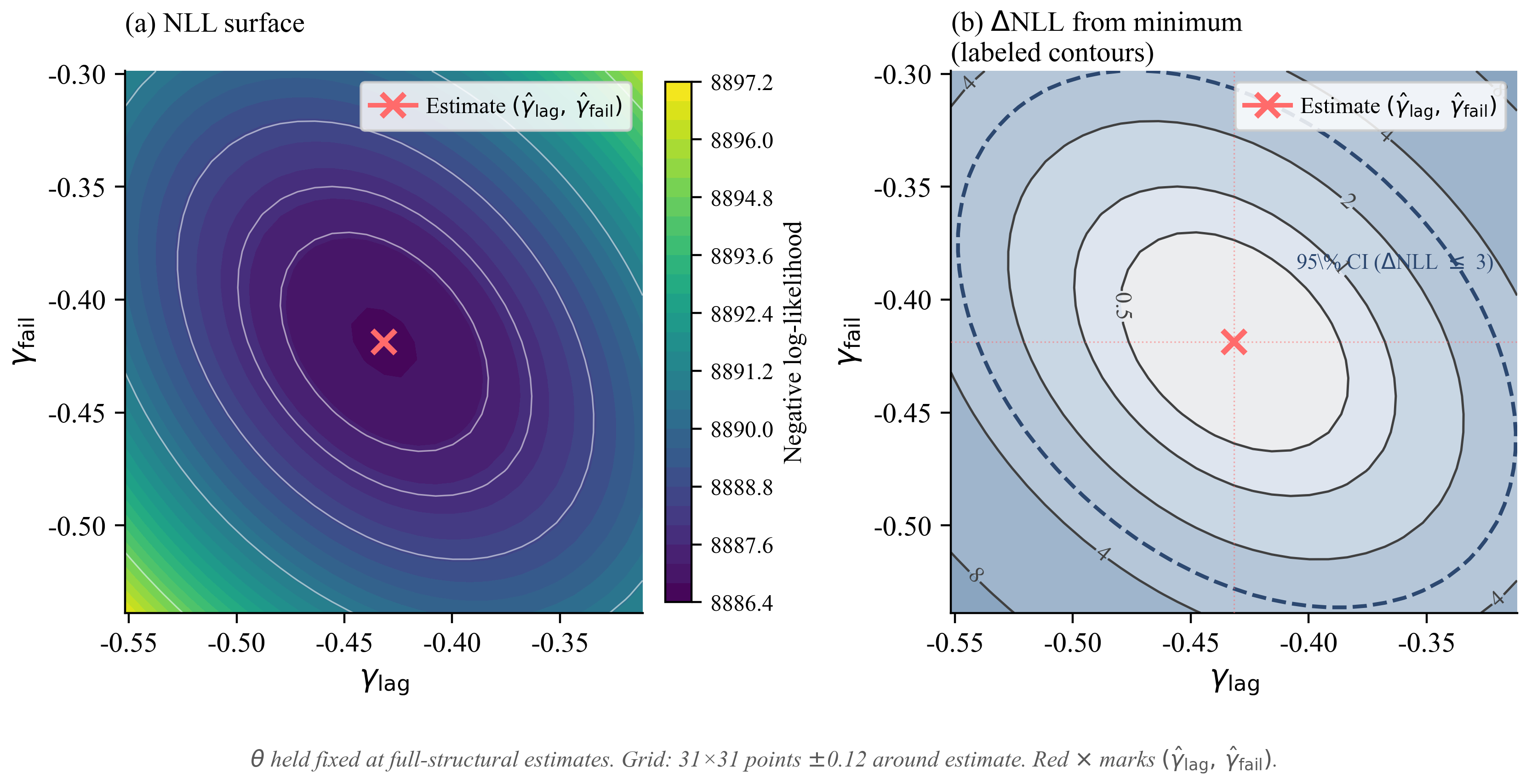}
\caption{Likelihood surface and contour map for the spatial interaction parameters $\gamma_{lag}$ and $\gamma_{fail}$.}
\label{fig:likelihood_surface}
\end{figure}

Panel (a) shows the negative log-likelihood surface evaluated on a grid around the estimated parameters. The likelihood exhibits a well-defined minimum near the estimated values, indicating stable identification of the spatial interaction parameters.

Panel (b) reports contours of the change in negative log-likelihood
relative to the minimum. The dashed contour corresponds to the
95\% likelihood region under the $\chi^2(2)$ approximation. The estimated parameter vector lies near the center of this region, confirming that the estimator has located the global optimum and that the likelihood surface is locally well behaved.

The elliptical shape of the likelihood contours indicates moderate
correlation between the two spatial interaction parameters. This pattern is expected because both parameters capture related forms of local maintenance coordination.

\subsection{Mechanisms of Local Coordination}

To assess the mechanisms underlying spatial coordination,
Table~\ref{tab:structural_channel_compare} reports restricted structural models in which only one interaction channel is active at a time. The lag-only specification sets $\gamma_{fail}=0$, while the fail-only specification sets $\gamma_{lag}=0$. Both models are estimated under the full structural likelihood.

\begin{table}[!h]
\centering
\caption{Structural comparison of spatial channel specifications}
\label{tab:structural_channel_compare}
\begin{tabular}{lccc}
\toprule
Model & $\hat{\gamma}_{lag}$ & $\hat{\gamma}_{fail}$ & NLL \\
\midrule
Fail-only & 0.0000 & -0.4712 & 8919.05 \\
Lag-only  & -0.4799 & 0.0000 & 8915.63 \\
Two-gamma (full) & -0.4314 & -0.4184 & 8886.75 \\
\bottomrule
\end{tabular}
\end{table}
Both restricted specifications recover negative spatial coefficients of similar magnitude. The fail-only model yields $\hat{\gamma}{fail} = -0.4712$, while the lag-only model yields $\hat{\gamma}{lag} = -0.4799$. These results indicate that each interaction channel independently generates substantial coordination in replacement behavior.

In terms of model fit, the lag-only specification achieves a slightly
lower negative log-likelihood than the fail-only model
(8915.63 vs.\ 8919.05). While this suggests that sequential replacement
activity provides a marginally better explanation of the data, the
difference in fit is small, and the estimated coefficients are of nearly
identical magnitude.

These results indicate that spatial coordination operates through both channels. Replacement decisions respond not only to activity nearby failure conditions, but also to recent neighborhood replacement. Neither mechanism alone fully captures the observed patterns, and their similar magnitudes suggest that coordination arises through both sequential maintenance activity and contemporaneous local reliability signals.

Persistent environmental heterogeneity is absorbed through cage-specific deterioration effects and fixed group structure. The spatial interaction effects are therefore identified from within-cage variation in replacement incentives conditional on local operating environments. In particular, the lagged replacement effect cannot be explained solely by contemporaneous thermal shocks because it operates through prior-period maintenance activity within the local neighborhood.

\subsection{Nonlinear Interaction Effects: Intensity of Local Activity}

The baseline specification models spatial interactions using binary indicators for local activity, indicating whether any neighboring unit experienced a recent replacement or failure. While this captures the presence of coordination effects, it imposes that the marginal impact of a single neighboring event is identical to that of multiple events. This restriction may be overly strong in environments with clustered maintenance activity.

To relax this restriction, interaction effects are allowed to vary with the intensity of local activity. Specifically, the binary indicators are replaced with binned measures distinguishing between one neighboring event and multiple neighboring events. The resulting flow utility adjustment is:

\begin{equation}
\Delta U_i =
\gamma_{r1} \cdot \mathbf{1}\{R_{i,t-1}=1\}
+ \gamma_{r2+} \cdot \mathbf{1}\{R_{i,t-1}\ge 2\}
+ \gamma_{f1} \cdot \mathbf{1}\{F_{i,t}=1\}
+ \gamma_{f2+} \cdot \mathbf{1}\{F_{i,t}\ge 2\},
\end{equation}

where $R_{i,t-1}$ denotes neighboring replacement activity in the previous period and $F_{i,t}$ denotes neighboring failures in the current period, with the latter defined excluding the focal unit.

\begin{table}[h!]
\centering
\caption{Interaction Effects by Intensity}
\label{tab:intensity}
\begin{tabular}{lc}
\toprule
 & Estimate \\
\midrule
$\gamma_{r1}$ (1 lagged replacement) & -0.252 \\
$\gamma_{r2+}$ ($\geq$2 lagged replacements) & -0.887 \\
$\gamma_{f1}$ (1 failure) & -0.375 \\
$\gamma_{f2+}$ ($\geq$2 failures) & -0.479 \\
\bottomrule
\end{tabular}
\end{table}

All coefficients are negative, indicating that both neighboring replacements and failures increase the relative payoff to replacement. The estimates exhibit a clear monotone pattern in interaction intensity: the effect of two or more lagged neighbor replacements is substantially larger in magnitude than that of a single replacement, and a similar pattern holds for failures, albeit more moderately.

The stronger gradient for lagged replacement activity relative to failures suggests that realized maintenance activity provides a sharper coordination signal than contemporaneous failure shocks. Observed replacements directly reveal active intervention and potential cost-sharing opportunities, while failures provide a noisier signal of underlying operating conditions.

This difference is reinforced by persistent heterogeneity in operating environments. Because cage conditions remain stable over time, lagged replacement activity partly reflects enduring local deterioration patterns, while contemporaneous failures capture more transitory shocks. Lagged replacement activity therefore contains stronger information about local operating conditions and potential coordination opportunities.

Allowing for intensity-dependent interaction effects improves model fit relative to the binary specification. The negative log-likelihood decreases by approximately 31 points, indicating that the additional flexibility captures systematic variation in replacement behavior. At the same time, the core structural parameters governing age dependence and failure effects remain stable, implying that the nonlinear interaction terms refine rather than overturn the baseline economic structure.

These results show that spatial coordination is not purely binary but
increases with the intensity of local activity. However, the modest
improvement in fit indicates that the binary specification captures the
primary margin of interaction.

The intensity effects point to nonlinear coordination mechanisms,
consistent with threshold or fixed-cost structures in maintenance
operations. Accordingly, the binary model is retained as the baseline,
with the intensity specification serving as a robustness check.

As a further robustness check, a more flexible specification is
estimated that distinguishes between one, two, and three or more
neighboring events. Despite the additional parameters, this specification
delivers only modest improvements in fit and yields less stable
estimates, with little structure beyond the monotonic pattern identified
above.

The estimated effects remain increasing in intensity but are not
precisely distinguished across higher bins, indicating that the
difference between two and three or more neighboring events is weakly
identified. This favors the more parsimonious specification that groups
higher levels of activity into a single category.

\subsection{Sensitivity to Discounting}

This section examines the sensitivity of the full structural estimates
to the choice of discount factor $\beta$. The baseline specification
sets $\beta = 0.90$, consistent with a moderate degree of forward-looking
behavior at a monthly frequency. To assess robustness, we re-estimate the model at $\beta = 0.85$ and $\beta = 0.95$. Table~\ref{tab:beta_robustness} reports the resulting parameter
estimates. 

\begin{table}[!h]
\centering
\caption{Sensitivity of Full Structural Estimates to Discount Factor $\beta$}
\label{tab:beta_robustness}
\begin{tabular}{lccc}
\toprule
 & $\beta=0.85$ & $\beta=0.90$ & $\beta=0.95$ \\
\midrule
NLL
& 8836.49 & 8886.75 & 8961.62 \\

\midrule
$\gamma_{\text{lag}}$ 
& -0.4306 & -0.4314 & -0.4343 \\

$\gamma_{\text{fail}}$ 
& -0.4168 & -0.4184 & -0.4189 \\

\midrule
Age $\times$ Cage 0 ($\theta_{\text{0}}$) 
& -0.0082 & -0.0060 & -0.0040 \\

Age $\times$ Cage 1 ($\theta_{\text{1}}$)
& -0.0120 & -0.0085 & -0.0052 \\

Age $\times$ Cage 2 ($\theta_{\text{2}}$) 
& -0.0254 & -0.0183 & -0.0118 \\

Failure Event ($\theta_{\text{fail}}$) 
& -8.7732 & -8.7453 & -8.7149 \\

Replacement Cost ($\theta_{\text{RC}}$) 
& -9.1214 & -9.3352 & -9.7928 \\
\bottomrule
\end{tabular}
\end{table}

The results exhibit two key patterns. First, the estimated spatial interaction parameters are highly stable across discount factors. Across all specifications, $\gamma_{\text{lag}}$ and $\gamma_{\text{fail}}$ remain close to $-0.43$ and $-0.42$,
respectively. This stability indicates that
the estimated coordination effects are not sensitive to assumptions
about intertemporal discounting and are primarily identified from
cross-sectional and short-run variation in the data.

Second, the own-state utility parameters adjust systematically with the
discount factor. As $\beta$ increases, the magnitude of the age-related
cost parameters declines, while the replacement cost parameter becomes
more negative.

The likelihood values vary across specifications, with lower values at higher discount factors, reflecting differences in overall model fit. However, these changes do not materially affect the estimated spatial interaction effects.

These findings indicate that the core substantive finding
of the paper—namely, the presence of economically meaningful spatial coordination in replacement decisions—is robust to reasonable variation in the discount factor. The choice of $\beta = 0.90$ therefore provides a convenient baseline without driving the main conclusions.


\section{Counterfactual Analysis}
\label{sec:cf_framework}

This section describes the simulation framework used to evaluate the
economic implications of spatial coordination. The analysis compares
predicted system outcomes across alternative models while holding the
underlying stochastic environment fixed.

\subsection{Simulation Framework}

For each estimated model, replacement behavior is characterized by a
policy function
\[
\pi(d \mid \st) = P(d=1 \mid \st),
\]
obtained as the fixed point of the model-specific Bellman equation.

Counterfactual outcomes are generated through forward simulation. Starting from an initial state distribution, the system evolves according to the estimated policy function and transition process. In each period, replacement decisions are drawn from the model-implied conditional choice probabilities, and states are updated using the estimated transition kernel.

The transition process is held fixed across simulations. Age evolves deterministically conditional on the replacement decision, and failure realizations follow the empirical hazard process estimated from the data. This ensures that differences across counterfactual scenarios are driven by differences in decision rules rather than changes in the stochastic environment.

Simulations are initialized from the empirical state distribution observed in the data, including age, failure status, and cage position. The simulation horizon is set to T = 36 periods, corresponding to a three-year analysis window.

Reported outcomes reflect the evolution of the system under each model’s
policy function over this fixed horizon.

\subsection{Replacement Dynamics Under Spatial Coordination}

Before turning to policy counterfactuals, it is useful to compare the dynamic implications of the estimated models. Although the spatial model improves in-sample fit relative to the baseline specification, the more important distinction is behavioral: the two models generate systematically different replacement patterns because they respond differently to local operating conditions.

The primary comparison is between two specifications:

\begin{enumerate}
    \item \textit{Baseline model:} Replacement decisions depend
    only on age, failure status, and cage position, with no localized
    interaction effects.

    \item \textit{Structural spatial model:} Replacement decisions additionally depend on local replacement and failure activity within the same cage.
\end{enumerate}

Both models share the same transition process and differ only in the policy function governing replacement behavior. Differences in simulated outcomes therefore reflect localized interaction effects in the structural decision rule.

The presence of spatial interaction effects changes the timing of
replacement decisions. In the baseline model, replacement occurs when
the expected continuation value of operating an aging unit falls below
the replacement value, driven primarily by age and failure risk. In the
spatial model, this decision margin is shifted by local conditions.

Lagged local replacement activity lowers the continuation value of operating a unit, increasing the likelihood of replacement in periods following local maintenance activity. Similarly, neighboring failures increase replacement incentives by revealing localized reliability deterioration within the operating environment.

Figure~\ref{fig:replacement_paths} compares simulated replacement
dynamics across the two models. 

\begin{figure}[H]
    \centering
    \includegraphics[width=\textwidth]{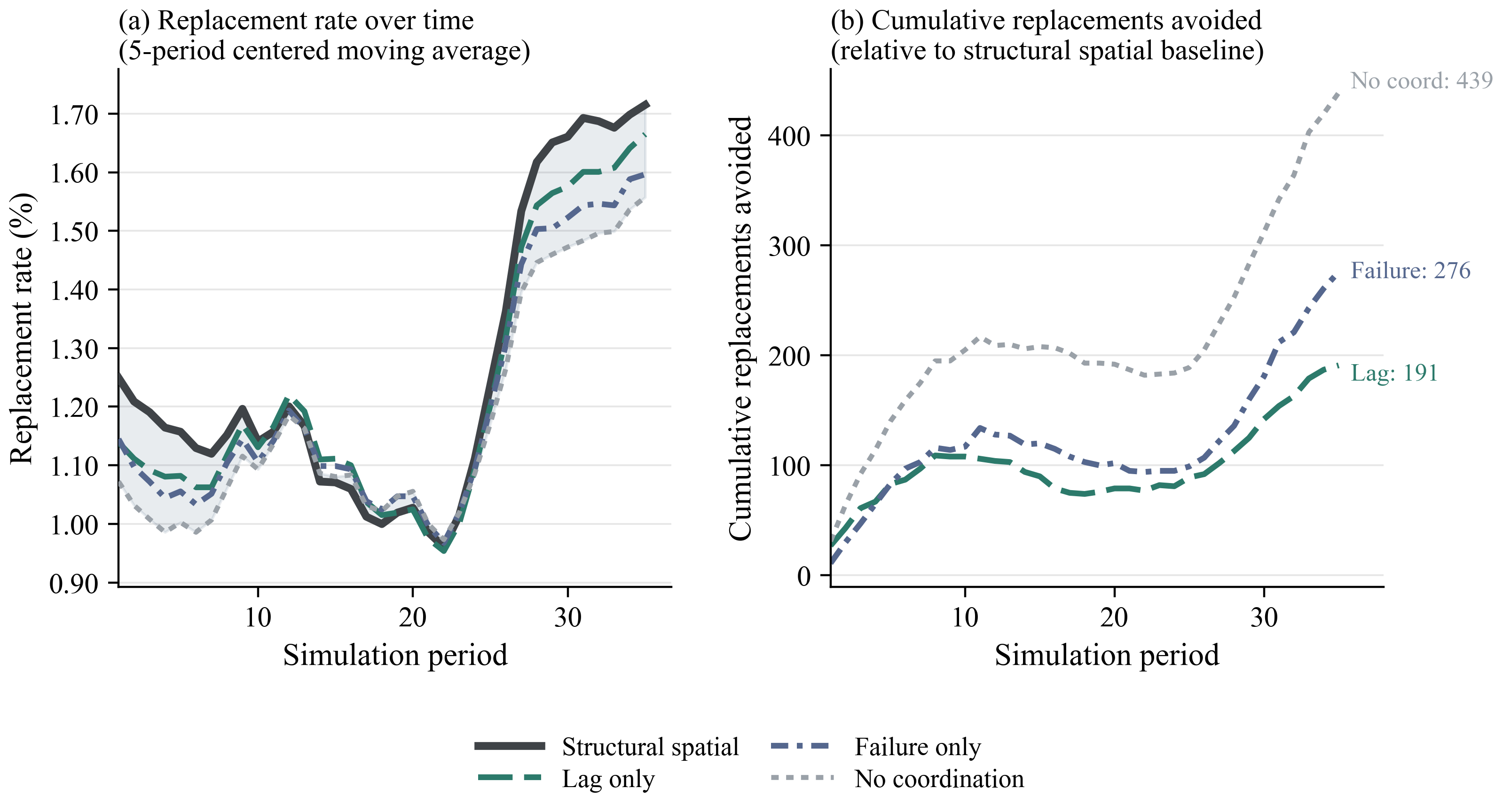}
    \caption{Replacement dynamics across the baseline and structural spatial models.}
    \label{fig:replacement_paths}
\end{figure}

The differences in replacement timing are reflected in the implied replacement hazard rates. In the baseline model, the hazard rate increases smoothly with age, reflecting gradual deterioration and increasing failure risk.

The spatial model modifies this pattern by introducing shifts in the hazard rate conditional on local activity. Replacement hazards increase discretely following periods of neighbor replacement or failure, creating deviations from the smooth age profile implied by the baseline model. These shifts generate localized spikes in replacement activity that are not captured by age-based deterioration alone.

A key implication of the spatial model is the clustering of maintenance activity within groups. Because incentives respond to local conditions, interventions tend to occur in sequences rather than in isolation. This clustering arises endogenously from the interaction structure of the model rather than from exogenous shocks.
\subsection{Coordination Channel Counterfactuals}

To isolate the contribution of each coordination mechanism, counterfactual
simulations selectively shut down the corresponding interaction terms in
the spatial model while holding all other structural parameters and the
transition process fixed.

\begin{itemize}
    \item \textit{No coordination:} Set
    $(\gamma_{lag}, \gamma_{fail}) = (0,0)$.

    \item \textit{Lag-only coordination:} Set $\gamma_{fail} = 0$.

    \item \textit{Failure-only coordination:} Set $\gamma_{lag} = 0$.
\end{itemize}

Because the transition environment is unchanged across simulations,
differences in outcomes are driven entirely by changes in the structural
decision rule.

Each simulation tracks replacement rates, failure rates, and mean equipment age, along with cumulative maintenance activity, cumulative failures, and total discounted costs. Total discounted costs are computed as the expected present value of replacement and failure costs under the model-implied policy, capturing the trade-off between preventive replacement and realized failures.

Figure~\ref{fig:decomp_counts} reports the decomposition of replacement dynamics across coordination mechanisms.

\begin{figure}[!h]
    \centering
    \includegraphics[width=0.9\textwidth]{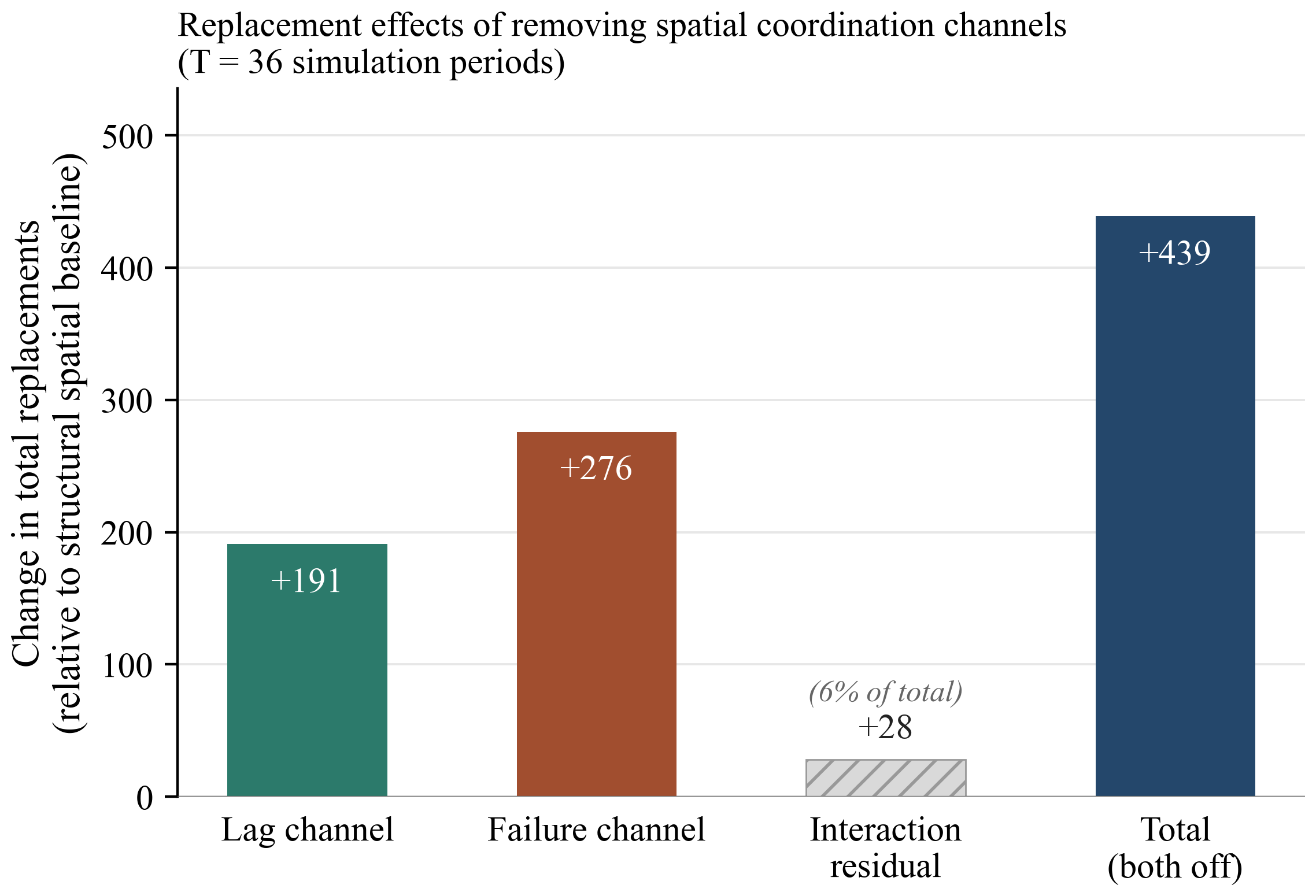}
    \caption{Decomposition of replacement dynamics under counterfactual restrictions on spatial interaction parameters.}
    \label{fig:decomp_counts}
\end{figure}

Removing the lagged replacement mechanism reduces total replacements by 191 units relative to the full spatial model, while removing the failure mechanism reduces replacements by 276 units. The larger failure effect suggests that local reliability shocks are the primary source of coordinated replacement behavior in the Titan system.

The interaction residual is comparatively small, accounting for only 6\% of the total effect. This indicates that the two mechanisms operate largely independently and that most coordinated replacement behavior is driven by additive responses to local maintenance activity and nearby failures rather than by higher-order interaction effects.

\subsection{Total Discounted Costs}

To evaluate the economic implications of spatial coordination, it is useful to distinguish between two comparisons: (i) differences across models and (ii) counterfactual changes within the structural spatial model itself. Costs are normalized using the Titan-era NVIDIA K20X GPU MSRP of \$7,699 to provide an interpretable mapping from utility units into dollar-denominated comparisons. The resulting values should therefore be interpreted as illustrative operational cost comparisons rather than exact accounting measures.

Relative to the non-spatial baseline specification, the fully structural spatial model achieves lower total discounted costs through more efficient replacement timing and a lower incidence of failures. By incorporating localized operating conditions and shared maintenance environments, the spatial model captures coordination effects that improve system performance at empirically observed levels of interaction.

Figure~5 shows that the relationship between coordination and system performance is not monotonic. The structural model improves replacement timing and reduces failures relative to the non-spatial baseline by incorporating local operating conditions into replacement decisions. Coordination therefore improves efficiency at empirically observed levels of interaction.

However, the within-model counterfactuals show that stronger spillovers are not necessarily optimal. Eliminating coordination channels reduces costs relative to the spatial baseline, while stronger spillovers induce earlier and more clustered replacement activity that raises discounted costs when additional capital expenditures outweigh marginal reductions in failure risk. The results therefore suggest that the observed level of coordination reflects a balance between preventive maintenance incentives and replacement costs rather than maximal coordination.

\begin{figure}[H]
    \centering
    \includegraphics[width=0.9\textwidth]{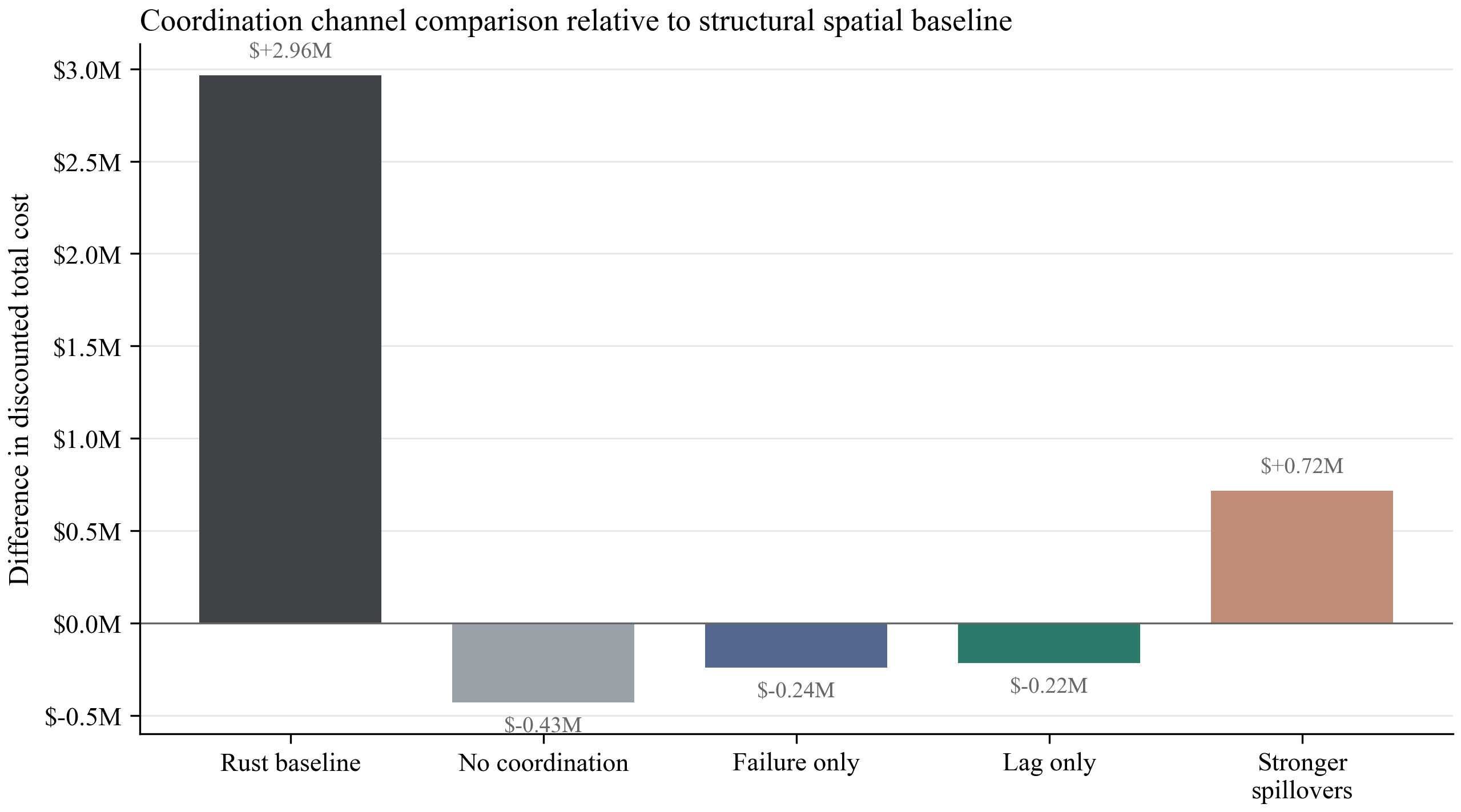}
    \caption{Total discounted cost differences under counterfactual restrictions on spatial coordination channels relative to the structural spatial baseline. Positive values indicate higher costs relative to the baseline.}
    \label{fig:channel_tvm}
\end{figure}

These differences are economically important because models that ignore spatial interactions attribute replacement behavior entirely to idiosyncratic aging and failure processes. By contrast, the spatial model captures how local operating conditions and neighboring activity jointly shape replacement incentives, generating materially different predictions for failure rates, equipment age, replacement clustering, and total costs.

Although the counterfactual cost comparisons are normalized using Titan-era GPU MSRP, the implied magnitudes become substantially larger in modern hyperscale computing environments. Modern AI accelerators can cost tens of thousands of dollars per unit, with deployments involving hundreds of thousands of GPUs. At these scales, even modest improvements in replacement timing or failure management can translate into operational differences measured in tens or hundreds of millions of dollars over the lifecycle of a large deployment. More broadly, the economic importance of localized coordination is likely to increase alongside the scale, capital intensity, and operational interdependence of modern infrastructure systems.

\section{Conclusion}

This paper develops an exact decomposition result for structural dynamic discrete choice models with localized interaction effects and fixed group structure. Under these conditions, the Bellman problem separates into independent group-level subproblems while preserving equilibrium value and policy functions exactly. The decomposition restores tractable estimation and allows inference and policy analysis for interaction-rich structural models at infrastructure scale.

Applying the framework to GPU replacement decisions in the Titan supercomputer reveals economically meaningful coordination in maintenance behavior. Replacement incentives not only respond to own-unit deterioration and failure risk, but also to neighboring maintenance activity and local reliability conditions. These interaction effects materially alter predicted replacement behavior, failure rates, and discounted operating costs relative to non-spatial models imposing conditional independence.

The counterfactual analysis shows that coordination has a non-monotonic effect on system performance. At empirically observed levels, coordination improves efficiency by incorporating localized operating conditions and shared maintenance environments into replacement decisions. However, stronger spillovers are not necessarily optimal because excessive responsiveness to local activity induces over-replacement and higher capital expenditures. The results therefore suggest that observed coordination behavior reflects a balance between preventive maintenance incentives and replacement costs rather than maximal coordination.

Although the empirical application focuses on Titan, the economic implications extend far beyond a single supercomputing environment. Modern hyperscale computing systems now involve deployments of hundreds of thousands of AI accelerators costing tens of thousands of dollars per unit, implying infrastructure investments measured in the billions of dollars. At these scales, even modest improvements in replacement timing, reliability management, or maintenance coordination can generate economically meaningful differences in system performance and capital expenditure.

Similar coordination problems arise in telecommunications networks, electricity infrastructure, transportation fleets, distributed energy systems, and other networked capital environments with persistent local operating conditions. Further, the framework provides a foundation for richer structural models of heterogeneous infrastructure systems. Future work could incorporate multiple interacting asset types, endogenous network formation, strategic interaction across decentralized operators, or equilibrium feedback effects such as congestion and endogenous pricing. 

This paper shows that stable localized dependence fundamentally changes the computational structure of large-scale infrastructure replacement problems. Exploiting this structure makes it possible to estimate dynamic models with interaction effects directly within the Bellman equation while preserving the original economic problem exactly. As infrastructure systems become more networked and capital intensive, estimating interaction-rich structural models at operational scale will become increasingly important for infrastructure planning and reliability management.

\newpage

\appendix

\section{Additional Derivations}

\subsection{Matrix Illustration of Exact Decomposition}
\label{app:matrix-decomposition}

Consider two fixed groups, indexed by $g \in \{0,1\}$, each with two
states. For action $d$, suppose the controlled transition matrix is
block diagonal:
\[
P_d
=
\begin{bmatrix}
P_d^{(0)} & 0 \\
0 & P_d^{(1)}
\end{bmatrix}
=
\begin{bmatrix}
p_{11}^{(0)} & p_{12}^{(0)} & 0 & 0 \\
p_{21}^{(0)} & p_{22}^{(0)} & 0 & 0 \\
0 & 0 & p_{11}^{(1)} & p_{12}^{(1)} \\
0 & 0 & p_{21}^{(1)} & p_{22}^{(1)}
\end{bmatrix}.
\]

Let the value vector be stacked by group:
\[
V
=
\begin{bmatrix}
V_1^{(0)} \\
V_2^{(0)} \\
V_1^{(1)} \\
V_2^{(1)}
\end{bmatrix}
=
\begin{bmatrix}
V^{(0)} \\
V^{(1)}
\end{bmatrix}.
\]

Then the expected continuation value is
\[
P_d V
=
\begin{bmatrix}
p_{11}^{(0)}V_1^{(0)}
+
p_{12}^{(0)}V_2^{(0)}
+
0\cdot V_1^{(1)}
+
0\cdot V_2^{(1)}
\\[0.5em]
p_{21}^{(0)}V_1^{(0)}
+
p_{22}^{(0)}V_2^{(0)}
+
0\cdot V_1^{(1)}
+
0\cdot V_2^{(1)}
\\[0.5em]
0\cdot V_1^{(0)}
+
0\cdot V_2^{(0)}
+
p_{11}^{(1)}V_1^{(1)}
+
p_{12}^{(1)}V_2^{(1)}
\\[0.5em]
0\cdot V_1^{(0)}
+
0\cdot V_2^{(0)}
+
p_{21}^{(1)}V_1^{(1)}
+
p_{22}^{(1)}V_2^{(1)}
\end{bmatrix}.
\]
\[
P_d V
=
\begin{bmatrix}
p_{11}^{(0)}V_1^{(0)} + p_{12}^{(0)}V_2^{(0)} \\
p_{21}^{(0)}V_1^{(0)} + p_{22}^{(0)}V_2^{(0)} \\
p_{11}^{(1)}V_1^{(1)} + p_{12}^{(1)}V_2^{(1)} \\
p_{21}^{(1)}V_1^{(1)} + p_{22}^{(1)}V_2^{(1)}
\end{bmatrix}.
\]

Equivalently, this can be computed group by group:
\[
P_d^{(0)}V^{(0)}
=
\begin{bmatrix}
p_{11}^{(0)} & p_{12}^{(0)} \\
p_{21}^{(0)} & p_{22}^{(0)}
\end{bmatrix}
\begin{bmatrix}
V_1^{(0)} \\
V_2^{(0)}
\end{bmatrix},
\qquad
P_d^{(1)}V^{(1)}
=
\begin{bmatrix}
p_{11}^{(1)} & p_{12}^{(1)} \\
p_{21}^{(1)} & p_{22}^{(1)}
\end{bmatrix}
\begin{bmatrix}
V_1^{(1)} \\
V_2^{(1)}
\end{bmatrix}.
\]

Thus,
\[
P_d V
=
\begin{bmatrix}
P_d^{(0)}V^{(0)} \\
P_d^{(1)}V^{(1)}
\end{bmatrix}.
\]

The decomposed computation is algebraically identical to the joint computation. The only omitted operations involve off-diagonal zero blocks implied by fixed group membership and localized transition dependence.

\bibliographystyle{plainnat}
\bibliography{references}

\end{document}